\documentclass[12pt]{amsart}
\usepackage{hyperref}
\usepackage{pdfsync}
\usepackage[latin1]{inputenc}
\usepackage{amsmath,amssymb}
\usepackage{colortbl}
\usepackage[english]{babel}
\usepackage [T1] {fontenc}
\newtheorem{thm}{Theorem}[section]
\newtheorem{cor}[thm]{Corollary}
\newtheorem{lem}[thm]{Lemma}

\theoremstyle{definition}
\newtheorem{defn}[thm]{Definition}
\newtheorem{rem}[thm]{Remark}
\newtheorem{ex}[thm]{Example}

 \numberwithin{equation}{section}
\newcommand\A{\mathcal{A}}

\newcommand\E{\mathcal{E}}
\newcommand\F{\mathcal{F}}
\newcommand\N{\mathbb{N}}

\newcommand\R{{\mathbb R}}
\newcommand\C{\mathbb{C}}
\newcommand\K{\mathcal{K}}

\newcommand\M{\mathcal{M}}

\newcommand\calH{\mathcal{H}}
\newcommand\J{\mathcal{J}}

\begin{document}
\title[Existence, degeneracy and stability of ground states]{Existence, degeneracy and stability of ground states by logarithmic Sobolev inequalities\\ on Clifford algebras.}
\author{Fabio E.G. Cipriani}
\address{(F.E.G.C.) Politecnico di Milano, Dipartimento di Matematica, piazza Leonardo da Vinci 32, 20133 Milano, Italy.}
\email{fabio.cipriani@polimi.it}
\footnote{This work has been supported by Laboratoire Ypatia des Sciences Math\'ematiques C.N.R.S. France - Laboratorio Ypatia delle Scienze Matematiche I.N.D.A.M. Italy (LYSM).}
\footnote{Author states no conflict of interest.}
\date{\today}
% ----------------------------------------------------------------
\begin{abstract}
We prove existence and finite degeneracy of ground states of energy forms satisfying logarithmic Sobolev inequalities with respect to non vacuum states of Clifford algebras. We then derive the stability of the ground state with respect to certain unbounded perturbations of the energy form. Finally, we show how this provides an infinitesimal approach to existence and uniqueness of the ground state of Hamiltonians considered by L. Gross in QFT, describing spin $1/2$ Dirac particles subject to interactions with an external scalar field.
\end{abstract}
\maketitle
\tableofcontents

\section{Introduction}
In \cite{13} L. Gross proved existence and finite degeneracy of ground states for energy forms generating positivity preserving, hyperconctractive semigroups, acting on the space $L^2(M,\tau)$ of {\it regular probability gage spaces}, i.e. von Neumann algebras $M$ endowed with normal trace states $\tau$. The result, derived from an infinite dimensional, noncommutative extension of the classical Perron-Frobenious theorem, is then applied to establish existence and finite degeneracy of ground states of Hamiltonians describing Fermions subject to interactions, in constructive Quantum Field Theory. In these applications $M$ is the Clifford algebra $Cl(\calH,\J)$ of the Hilbert space with conjugation $(\calH,\J)$ of an electron-positron pair and $\tau$ is the bare vacuum state. In \cite{13}, hypercontractivity is first proved for the semigroup $e^{-tN}$ generated by the number operator $N$, representing the second quantized Hamiltonian of an electron-positron free quantum field, and then it is verified for Hamiltonians involving interactions. In a subsequent work \cite{14}, L. Gross proved that the hypercontractivity of the semigroup $e^{-tN}$ actually comes from a logarithmic Sobolev inequality satisfied by its generator $N$ (see also \cite{5}, \cite{18}). The methods used by L. Gross on general von Neumann algebras $M$ and, in particular, on the Clifford algebra of a separated Hilbert space with involutions, partially rely on the hyperfinitness of $M$, i.e. on the possibility to approximate $M$ by finite dimensional algebras. The need of this assumption was subsequently eliminated by W. Wils in \cite{24}.
\vskip0.2truecm\noindent
The purpose of the present work is to show that existence and finite degeneracy of ground states of energy forms, generating strongly continuous, positivity preserving semigroups on the Hilbert space $L^2(M,\tau)$ of a Clifford algebra $M:=Cl(\calH,\J)$, can be derived directly from the existence of a logarithmic Sobolev inequality.
\vskip0.2truecm\noindent
To this end, in Section 2 we recall the notions of energy forms and Dirichlet forms on the Clifford algebra, emphasising the standard von Neumann structure of the latter (\cite{6}, \cite{7}). We then generalise, in Section 3, the classical tools of {\it convergence in probability} and {\it uniform integrability} for normal states in noncommutative integration theory (see \cite{11}, \cite{22}, \cite{23}, \cite{19}) and then we establish that sets of states having {\it uniformly bounded relative entropy}, with respect to a reference normal state, are uniformly integrable. In Section 4, we apply this approach to show the {\it stability of the ground state} under certain unbounded perturbations of the energy form. In Section 5, we outline how these tools may used to give an alternative, infinitesimal approach based on the logarithmic Sobolev inequality in \cite{14} , to the L. Gross' result in \cite{13} concerning the existence and uniqueness of the ground state of certain physical Hamiltonians for systems of spin 1/2 Dirac particles.

\section{Energy and Dirichlet forms on Clifford algebras}

In this work we denote by $M$ the Clifford algebra $Cl(\calH,\J)$ of a complex Hilbert space $\calH$ with conjugation $\J$ and we refer to Section 4 in \cite{13} for the concerning material.

\subsection{Standard form of Clifford algebras}
The Clifford algebra $M$ is defined as the weakly closed algebra of operators on the Fermi-Fock space $\Lambda(\calH)$ generated by the bounded operators
\[
B_x:=C_x+A_{Jx}\qquad x\in\calH,
\]
where $C_x$ and $A_y:=C^*_y$ denote the familiar {\it creation and annihilation operators} on $\Lambda(\calH)$, associated to vectors $x,y\in\calH$. It is a von Neumann algebra whose generators satisfy the characteristic commutation relations
\[
B_xB_y+B_yB_x=2(\J x|y)\qquad x,y\in\calH
\]
and the involution in $M$ corresponds to the conjugation by $B_x^*=B_{\J x}$. The {\it bare vaccum} vector $\Omega\in\Lambda^0(\calH)\subset\Lambda(\calH)$, represented by $1\in\C=\Lambda^0(\calH)$, provides a faithful, normal state
\[
\tau:M\to\C\qquad \tau(a):=(\Omega|a\Omega)_{\Lambda(\calH)}\qquad a\in M,
\]
which is a trace as it verifies the relations
\[
\tau(ab)=\tau(ba)\qquad a,b\in M.
\]
The GNS space $L^2(M,\tau)$ is the Hilbert space completion of $M$ under the norm defined by the sesquilinear form $M\ni a,b\mapsto\tau(a^*b)=(a\Omega,b\Omega)_{\Lambda(\mathcal{H})}$ and in it the unit $1_M$ of the algebra
provides the norm one vector $\xi_\tau\in L^2(M,\tau)$ representing the trace
\[
\tau(a)=(\xi_\tau|a\xi_\tau)\qquad a\in M.
\]
The left regular representation of the involutive algebra $M$ extends to its GNS representation $L:M\to B(L^2(M,\tau))$. Identifying $M$ with its image $L(M)$, the action $L(a)\xi$ of $a\in M$ on $\xi\in L^2(M,\tau)$ is denoted by $a\xi\in L^2(M,\tau)$.
The vector representing the trace is cyclic, $\overline{M\xi_\tau}=L^2(M,\tau)$, and the map $L^2(M,\tau)\ni a\xi_\tau\mapsto a\Omega\in\Lambda(\calH)$ extends to a unitary isomorphism $D:L^2(M,\tau)\to\Lambda(\calH)$ between the GNS space and the Fermi-Fock space, called {\it duality transform}, which transform the cyclic trace vector into the bare vacuum $D\xi_\tau=\Omega$.\\
The commutant von Neumann algebra
\[
M':=\{b\in B(L^2(M,\tau):ab=ba\,\,\forall a\in M\}
\]
is generated by the extension $R:M\to B(L^2(M,\tau))$ of the right regular representation of $M$.\\
Under the duality transform, the left and right actions of the generators appear as
\[
DL_{B_x}D^{-1}=C_x+A_{\J x}, \qquad DR_{B_x}D^{-1}=(C_x-A_{\J x})S,\qquad x\in\calH,
\]
where $S:=\Gamma(-I_\calH)$ is the symmetry which is the second quantization of the bounded operator $-I_\calH$ on $\calH$ (equivalently one has $D^{-1}S D(B_x)=B_{-x}$ for any $x\in\calH$ on $L^2(M,\tau)$). \vskip0.2truecm\noindent
The positive part $M_+:=\{a^*a\in M:A\in M\}$ of the Clifford algebra determines the subset
\[
L^2_+(M,\tau):=\overline{M_+\xi_\tau}
\]
of the GNS space, which  is {\it self-polar} as it coincides $L^2_+(M,\tau)=L^2_+(M,\tau)^\circ$ with its polar set
\[
L^2_+(M,\tau)^\circ:=\{\xi\in L^2(M,\tau):(\eta|\xi)\ge 0,\,\text{for all}\,\,\eta\in L^2_+(M,\tau)\}.
\]
In particular, it is a {\it closed, convex cone} which defines a real subspace
\[
L^2_\R(M,\tau):=L^2_+(M,\tau)-L^2_+(M,\tau)
\]
and a real structure by which $L^2(M,\tau)$ is the complexification of $L^2_\R(M,\tau)$. The {\it conjugation}
\[
J:L^2(M,\tau)\to L^2(M,\tau)\qquad J(\xi+i\eta)=\xi-i\eta,\qquad \xi,\eta\in L^2_\R(M,\tau),
\]
extend the involution of $M$ in the sense that
\[
J(a\xi_\tau)=a^*\xi_\tau\qquad a\in M.
\]
The {\it positive part} $\xi_+\in L^2_+(M,\tau)$ of a real vector $\xi\in L^2_\R(M,\tau)$, is defined as its projection onto the closed, convex positive cone $L^2_+(M,\tau)$. The {\it negative part}
$\xi_-:=\xi_+-\xi$ of $\xi$ is a {\it positive} vector {\it orthogonal} to the positive part $\xi_+$ and such that $\xi=\xi_+-\xi_-$ appears as the difference of positive, orthogonal vectors and its {\it modulus} is defined as $|\xi|:=\xi_++\xi_-$.
%The space $L^2(M,\tau)$, ordered by the cone $L^2_+(M,\tau)$, is a Riesz space only when ${\rm dim}\,(\calH)=1$.
\vskip0.1truecm\noindent
The triple $(M,L^2(M,\tau),L^2(M,\tau))$ is the {\it standard form or standard representation} of $M$, in the sense that it has the following characterising properties
\begin{itemize}
\item $JMJ=M'$
\item $JaJ=a^*,\qquad a\in M\cap M'$
\item $J\xi=\xi,\qquad \xi\in L^2_+(M,\tau)$
\item $aJaJ(L^2_+(M,\tau))\subseteq L^2_+(M,\tau),\qquad a\in\M$
\end{itemize}
where $M':=\{b\in B(L^2(M,\tau)):ab=ba,\,\, a\in M\}$ being the commutant von Neumann algebra of $M$. The entire von Neumann algebra $M$ can be reconstructed by the ordered Hilbert space $(L^2(M,\tau),L^2_+(M,\tau))$. The Clifford algebra is actually a factor in the sense that $M\cap M'= \C\cdot 1_M$.\\
The standard form of general von Neumann algebras were discovered by H. Araki \cite{2} and A. Connes \cite{9} for $\sigma$-finite von Neumann algebras and by U. Haagerup \cite{17} in the general case. A detailed exposition can be found in Section 2.5.4 of \cite{4} and a discussion of its application to physics can be found in Sections 5.1, 5.2 of \cite{16}.
\vskip0.2truecm\noindent
A fundamental property of the standard representation of a von Neumann algebra is that any normal, positive functional $\psi\in M_{*+}$ is represented by a {\it unique} vector $\xi_\psi\in L^2_+(M,\tau)$ as
\[
\psi(a)=(\xi_\psi|a\xi_\psi)\qquad a\in M.
\]
in particular, each normal state is represented by a unique, positive vectors of norm one. This one-to-one correspondence between normal states and positive vectors of norm one is the generalization of the operation to take the square root of a positive measurable function of unit integral on a measurable Lebesgue space $(X,\nu)$ where the (commutative) von Neumann algebra is $L^\infty(X,\nu)$, the space of normal functionals is $L^1(X,\nu)$ and the standard form Hilbert space is $L^2(X,\nu)$ acted on by $L^\infty(X,\nu)$ through pointwise multiplication.\\
In case of the type I von Neumann algebra $B(k)$ of all bounded operators on a separable Hilbert space $k$, the Hilbert space of the standard representation is the space $L^2(k)$ of Hilbert-Schmidt operators on $k$ and the action of the elements of $B(k)$ is given by the left composition. Normal, positive functionals on $B(k)$ can be identified with the trace class positive operators $\rho$ on $k$ and the correspondence mentioned above reduce to take the spectral square root $\sqrt{\rho}\in L^2(k)$.\\
A normal, positive functional is {\it faithful}, i.e. $a\in M_+$ and $\psi(a)=0$ imply $a=0$, if and only if the vector
$\xi_\psi$ is {\it strictly positive}, i.e. $\eta\in L^2_+(M,\tau)$ and $(\xi_\psi|\eta)=0$ imply $\eta=0$, or if $\overline{M_+\xi_\psi}=L^2_+(M,\tau)$.

\subsection{Energy and Dirichlet forms}

We shall be concerned with {\it energy forms} defined as the lower semicontinuous and lower bounded quadratic forms $\E:L^2(M,\tau)\to (-\infty,+\infty]$
\begin{itemize}
\item {\it finite on a dense subspace}, called the domain $\F:=\{\xi\in L^2(M,\tau): \E[\xi]<+\infty\}$
\item {\it real}, in the sense that $\E[J\xi]=\E[\xi]$ for all $\xi\in\F$, and satisfying the following
\item {\it first Beurling-Deny contraction property}
\[
\E[|\xi|]\le\E[\xi]\qquad J\xi=\xi\in\F.
\]
\end{itemize}
These quadratic forms are in one to one correspondence with the strongly continuous, symmetric semigroups on $L^2(M,\tau)$ which are
\begin{itemize}
\item {\it real}: $T_t(J\xi)=J(T_t\xi)$ for all $\xi\in L^2(M,\tau)$ and $t\ge 0$ and
\item {\it positivity preserving}: $T_t\xi\in L^2_+(M,\tau)$ for all $\xi\in L^2_+(M,\tau)$ and $t\ge 0$.
\end{itemize}
If $(H,D(H))$ denotes the self-adjoint, lower semibounded operator associated to the closed quadratic form $(\E,\F)$ on $L^2(M,\tau)$ and $\lambda_0:=\inf sp(H)$, one has the intertwining relations
\[
D(\sqrt{H})=\F,\quad \E[\xi]=\|\sqrt{H-\lambda_0}\xi\|^2+\lambda_0\|\xi\|^2_2\quad \xi\in \F,\qquad T_t=e^{-tH}\quad t\ge 0.
\]
In Quantum Mechanics, an energy form $\E$ represents the energy of a system whose Schr\"odinger operator is $H$.
\vskip0.2truecm\noindent
A related class of forms we shall be concerned with is the following one. Let $\psi\in M_{*+}$ be a fixed faithful state and $\xi_\psi\in L^2_+(M,\tau)$ its representative positive vector. A {\it Dirichlet form} on $L^2(M,\tau)$ {\it relatively to} $\psi$ is a lower semicontinuous and lower bounded quadratic functional $\E:L^2(M,\tau)\to (-\infty,+\infty]$
\begin{itemize}
\item {\it finite on a dense subspace} $\F:=\{\xi\in L^2(M,\tau): \E[\xi]<+\infty\}$, called the domain
\item {\it real}, $\E[J\xi]=\E[\xi]$ for all $\xi\in\F$, and satisfying the following contraction property
\item \[\text{\it Markovianity}\qquad \E[\xi\wedge \xi_\psi]\le\E[\xi],\qquad \forall\,\,J\xi=\xi\in\F.\]
\end{itemize}
Here, following Section 4 of \cite{6}, the vector $\xi\wedge\xi_\psi\in L^2(M,\tau)$ is the Hilbert space projection of $\xi\in L^2_\R(M,\tau)$ onto the closed and convex subset $\{\xi\in L^2_\R(M,\tau): \xi\le\xi_\psi\}$ of real vectors lying below $\xi_\psi$.  In case $\psi=\tau$ and when the vectors in $L^2(M,\tau)$ are thought as closed operators on $L^2(M,\tau)$, $\xi\wedge\xi_\tau\in L^2(M,\tau)$ is the functional calculus $f(\xi)$ by $f(s):=\min(s,1)$ for $s\in\R$ (see \cite{1}, \cite{21}, \cite{10}).
\vskip0.2truecm\noindent
Dirichlet forms are in one to one correspondence with {\it Markovian semigroups relatively to} $\psi$: these are the strongly continuous, symmetric semigroups $\{T_t:t\ge 0\}$ on $L^2(M,\tau)$ which are
\begin{itemize}
\item {\it real}: $T_t(\xi^*)=(T_t\xi)^*$ for all $\xi\in L^2(M,\tau)$ and $t\ge 0$
\item {\it Markovian}: $T_t\xi\le\xi_\psi$ for all $\xi=\xi^*\le\xi_\psi$ and $t\ge 0$.
\end{itemize}
In particular, Dirichlet forms are automatically nonnegative energy forms and Markovian semigroups are automatically
\begin{itemize}
\item {\it contractive}: $\|T_t\xi\|\le\|\xi\|$ for all $\xi\in L^2(M,\tau)$ and $t\ge 0$
\item {\it positivity preserving}: $T_t\xi\in L^2_+(M,\tau)$ for all $\xi\in L^2_+(M,\tau)$ and $t\ge 0$.
\end{itemize}
A condition ensuring that an energy form is a Dirichlet form with respect to a state $\psi\in M_{*+}$ is $\E[\xi_\psi]=0$ (in particular this requires $\xi_\psi\in\F$). In this situation, $\psi$ is called a {\it ground state}.
\begin{rem}
In Quantum Mechanics of spinless Bosons systems, a Dirichlet form represents the energy observable, in the so called {\it ground state representation}. The notion of Markovianity of an energy form
$(\E,\F)$ with respect to a state $\psi\in M_{*+}$ can thus be considered as a generalisation, for spin $1/2$ Fermions systems, of the ground state representation with respect to the ground state
$\psi$.
\end{rem}

\begin{ex}(Clifford-Dirichlet form of the Number operator)
Let $N$ be the number operator $d\Gamma(I_\calH)$ on the Fermi-Fock space $\Lambda(\calH)$, viewed as an operator on $L^2(M,\tau)$
\[
N:=D^{-1}d\Gamma(I_\calH)D.
\]
It is self-adjoint, positive with pure point spectrum given by $sp(N)=\N$. $\lambda_0=0$ is the only discrete eigenvalue and it is non degenerate with strictly positive eigenvector $\xi_0=\xi_\tau$. All others eigenvalues have countable multiplicity.\\
To describe its quadratic form, let $\{x_i\}_{i\in\N^*}\subset\calH$ be any orthonormal basis of real vectors $\J x_i=x_i$, $i\in\N^*$.
%Then, setting $\xi_{i_1,\cdots,i_n}:=B_{x_{i_1}}\cdots B_{x_{i_n}}\xi_\tau$, a real, orthonormal basis of $L^2(M,\tau)$ is given by
%\[
%\{\xi_0\}\cup \{\xi_{i_1,\cdots,i_n}\in L^2(M,\tau):i_1,\cdots,i_n\in\N^*,\,\,1\le i_1<\cdots <i_n,\,\, n\in\N^*\}.
%\]
Then, the number operator can be represented as
\[
N:=\sum_{i\in\N}a_i^*a_i
\]
where $a_i:=D^{-1}A_{x_i}D\in B(L^2(M,\tau))$ are the annihilation operators acting on $L^2(M,\tau)$.  The quadratic form of $N$, called {\it Clifford-Dirichlet form} in \cite{14}, then appears as
\[
\E_N[\xi]:=\|\sqrt{N}\xi\|^2=\sum_{i\in\N^*}\|a_i\xi\|^2.
\]
The semigroup $e^{-tN}$ on $L^2(M,\tau)$ is Markovian with respect to the trace $\tau$ as proved in Lemma 6.1 and 6.2 of \cite{13}. This result may be achieved showing that $\E_N$ is a Dirichlet form with respect to $\tau$. This is suggested by the observation of L. Gross (Section 2 in \cite{13}) that annihilation operators $a_i:L^2(M,\tau)\to L^2(M,\tau)$ are (bounded) derivations in the following sense:
\begin{itemize}
\item the space $\A:=M\xi_\tau\cap D(\sqrt{N})$ is a form core for $(\E_N,D(\sqrt{N}))$ and an Hilbert subalgebra of the Hilbert algebra $M\xi_\tau$, where
\[
(x\xi_\tau)(y\xi_\tau):=xy\xi_\tau,\qquad (x\xi_\tau)^*:=x^*\xi_\tau\qquad x,y\in M
\]
 \item the GNS space $L^2(M,\tau)$ is a $\A$-bimodule where the right action is the GNS one while the left action is the GNS left action twisted with the symmetry $S:=D^{-1}\Gamma(-I_\calH)D$ of
 $Cl(\calH,\J)$ (see Section 2.1)
\item the Leibniz rule holds true
\[
a_i(xy\xi_\tau)=(a_i(x\xi_\tau))y+(S(x)(a_i(y\xi_\tau))\qquad x,y\in M,\,\,i\in\N^*.
\]
\end{itemize}
Setting $\partial :=\oplus_{i=1}^\infty a_i$ one gets a derivation on $\A$ with values in the $\A$-bimodule $\K:=\oplus_{i=1}^\infty L^2(Cl(\calH,\J))$ by which $\E_N[\xi]=\|\partial\xi\|^2_\K$. Since $\partial$ is also closable, a general result (\cite{8}) implies that $\E_N$ is Dirichlet form with respect to $\tau$ (see \cite{7} for the details).
\end{ex}

\section{Tools of noncommutative integration}

In this section we introduce some tools of noncommutative integration theory, concerning convergence of normal functionals on von Neumann algebras.
\vskip0.2truecm\noindent
By the properties of the standard form, any $\psi\in M_{*+}$ can be represented by a unique positive vector $\xi_\psi\in L^2_+(M,\tau)$ as $\psi(x)=(\xi_\psi|x\xi_\psi)$ for any $x\in M$. Recall that
$\psi$ is faithful if and only if $\xi_\psi$ is strictly positive or if $M_+\xi_\psi$ is dense in the positive cone $L^2_+(M,\tau)$. In these cases $\xi_\psi$ is cyclic for $M$: $\overline{M\xi_\psi}=L^2(M,\tau)$. The positive cone $L^2_+(M,\tau)$ is the closure of $\{yJyJ\xi_\tau\in L^2(M,\tau): y\in M\}$ since $yJyJ\xi_\tau=yJy\xi_\tau=yy^*\xi_\tau\in L^2_+(M,\tau)$ for all $y\in M$ and $yJyJ\eta\in L^2_+(M,\tau)$ for any $y\in M$ and $\eta\in L^2_+(M,\tau)$.

\begin{lem}
Let $\varphi,\psi\in M_{*+}$ be positive functionals and assume $\psi$ to be faithful. Then a linear, densely defined, positive, symmetric, closable operator $(R_\varphi,D(R_\varphi))$ is given by
\[
D(R_\varphi):=M\xi_\psi\qquad R_\varphi(x\xi_\psi):=Jx^*\xi_\varphi\qquad x\in M.
\]
%and $\varphi(x)=\|R_\varphi\xi_\psi\|^2=(\xi_\psi|R^2_\varphi\xi_\psi)$ for all $x\in M$.
\end{lem}
\begin{proof}
The operator is densely defined since $\psi$ is faithful. It is symmetric since for all $x,y\in M$
\[
(y\xi_\psi|R_\varphi(x\xi_\psi))=(y\xi_\psi|Jx^*\xi_\varphi)=(x^*\xi_\varphi |Jy\xi_\psi)=(x^*\xi_\varphi |JyJ\xi_\psi)=(Jy^*J\xi_\varphi |x\xi_\psi)=(R_\varphi(y\xi_\varphi)|x\xi_\psi).
\]
Since $\xi_\varphi,\xi_\psi$ are positive, by the properties of the standard form, the positivity of $(R_\varphi,D(R_\varphi))$ follows from $(x\xi_\psi|R_\varphi(x\xi_\psi))=(x\xi_\varphi|Jx^*\xi_\varphi)=(\xi_\varphi|x^*Jx^*J\xi_\psi)\ge 0$ for all $x\in M$.
\end{proof}
The closure of the above operator, still denoted by $(R_\varphi,D(R_\varphi))$, is thus a densely defined, self-adjoint, positive operator on $L^2(M,\tau)$. It is part of the Connes' Radon-Nikodym cocycle \cite{2}, but for the purposes of the present work the solely properties we need of it are those described above.
\begin{ex}
In case $\psi=\tau$ (the situation considered in \cite{13}), for the {\it perturbed state}
\[
\varphi^h(x):=\tau(hx)\qquad x\in M,
\]
associated to a fixed, densely defined, self adjoint, positive, invertible, trace class operator $(h,D(h))$ on $L^2(M,\tau)$, affiliated to $M$ and such that $\tau(h)=1$, one has $\xi_\tau\in D(\sqrt{h})$ and $\xi_{\varphi^h}=\sqrt{h}\xi_\tau$.
%$\varphi_y(x):=\tau(y^*xy)$ for all $x\in M$ and some fixed $y\in M$, one has $\xi_{\varphi_y}=|y|\xi_\tau$ and $R_{\varphi_y}=|y|$
Since $(h,D(h))$ is affiliated to $M$, it commutes with all bounded operators $Jx^*J\in M'$ for any $x\in M$, so that
\[
R_{\varphi^h}(x\xi_\tau)=Jx^*J\xi_{\varphi^h}=Jx^*J\sqrt{h}\xi_\tau=\sqrt{h}Jx^*J\xi_\tau=\sqrt{h}Jx^*\xi_\tau=\sqrt{h}x\xi_{\tau}\qquad x\in M.
\]
Since $M\xi_\tau$ is a core both for $R_{\varphi^h}$ and $h$, we have $R_{\varphi^h}=\sqrt{h}$. The Radon-Nikodym derivative $d\varphi^h/d\tau$ is given by $R^2_{\varphi^h}=h$ and it is affiliated to $M$.
\end{ex}
The spectral measure of $R_\varphi$ will be indicated by $\mathbb{E}^\varphi$ so that $R_\varphi=\int_0^{+\infty}\mathbb{E}^\varphi (d\lambda)\,\,\lambda$.

%\subsection{Normal functionals vanishing in probability}
\subsection{Relative uniform integrability of states}

The following is meant to generalise {\it convergence in probability to zero} of sequences of integrable functions on probability spaces.
%It could formulated in terms of the I.E. Segal's convergence in measure but we prefer to formulate it in terms of states......

\begin{defn}[Relative vanishing for sequence of functionals]
A sequence $\varphi_n\in M_{*+}$ {\it vanishes relatively to a faithful $\psi\in M_{*+}$}, we also say that $\varphi_n\to 0$ {\it relatively to} $\psi$, if
\begin{equation}
\lim_n \mathbb{E}^{\varphi_n}(k,+\infty)\xi_\psi=0\qquad \forall\,\, k>0.
\end{equation}
\end{defn}

The following is a simple condition ensuring relative vanishing.

\begin{lem}
If the vectors $\xi_{\varphi_n}\in L^2_+(M)$, representing the functionals $\varphi_n\in M_{*+}$ satisfies
\begin{equation}
\lim_n (\xi_{\varphi_n}|\xi_\psi)=0,
\end{equation}
then $\varphi_n\to 0$ relatively to $\psi$.
\end{lem}

\begin{proof}
The assertion follows from
\[
\|\mathbb{E}^{\varphi_n}(k,+\infty)\xi_\psi\|_2^2=(\xi_\psi|\mathbb{E}^{\varphi_n}(k,+\infty)\xi_\psi)\le\frac{1}{k}\cdot (\xi_\psi|R_{\varphi_n} \xi_\psi)=\frac{1}{k}\cdot (\xi_\psi|\xi_{\varphi_n})\qquad k>0.
\]
\end{proof}

%\subsection{Relative uniform integrability of states}

The following generalises uniform integrability of random variables on a probability space.

\begin{defn}[Relative uniform integrability of normal functionals]
A family $\mathcal{C}\subset M_{*+}$ is said to be {\it relatively uniformly integrable with respect to a faithful state} $\psi\in M_{*+}$ if
\begin{equation}
\lim_{k\to+\infty}\sup_{\varphi\in\mathcal{C}} \|\mathbb{E}^{\varphi}(k,+\infty)\xi_\varphi\|_2^2=0
\end{equation}
or, equivalently, if
$
\lim_{k\to+\infty}\sup_{\varphi\in\mathcal{C}} (\xi_\varphi|\mathbb{E}^{\varphi}(k,+\infty)\xi_\varphi).
%=\lim_{k\to+\infty}\sup_{\varphi\in\mathcal{C}} (\xi_\psi|\mathbb{E}^{\varphi}(k,+\infty) R^2_\varphi\xi_\psi)=0
$
\end{defn}

Any finite set of positive functionals is relatively uniformly integrable with respect to any fixed faithful state. Any sequence of positive functionals, norm convergent to zero, is vanishing and relatively uniformly integrable with respect to any faithful $\psi\in M_{*+}$.

\subsection{Norm convergence to zero and uniform integrability}

The following is a noncommutative version of a classical result relating norm convergence in $L^1(X,m)$ to uniform integrability (see e.g. \cite{12}).

\begin{thm}
A sequence $\varphi_n\in M_{*+}$ uniformly integrable relatively to a faithful $\psi\in M_{*+}$ and such that
\[
\lim_n (\xi_\psi|\xi_{\varphi_n})=0,
\]
norm converges to zero: $\lim_n\|\varphi_n\|_{M_*}=0$.
\end{thm}

\begin{proof}
Consider the splitting
\[
\|\varphi_n\|_{M_{*+}}=(\xi_{\varphi_n}|\xi_{\varphi_n})=(\xi_{\varphi_n}|\mathbb{E}^{\varphi_n}(k,+\infty)\xi_{\varphi_n})
+(\xi_{\varphi_n}|\mathbb{E}^{\varphi_n}[0,k]\xi_{\varphi_n}).
\]
Fixing $\varepsilon>0$, by the uniform integrability hypothesis, there exists $k_\varepsilon>0$ by which the first term on the right hand side can be uniformly bounded as follows
\[
\sup_n (\xi_{\varphi_n}|\mathbb{E}^{\varphi_n}(k_\varepsilon,+\infty)\xi_{\varphi_n})<\varepsilon.
\]
Since for $\lambda\in (0,k_\varepsilon]$ one has $\lambda\le \sqrt{k_\varepsilon}\cdot\sqrt{\lambda}$, the second summand can be evaluated as follows
\[
\begin{split}
(\xi_{\varphi_n}|\mathbb{E}^{\varphi_n}[0,k_\varepsilon]\xi_{\varphi_n})
&=(R_{\varphi_n}\xi_\psi|\mathbb{E}^{\varphi_n}[0,k_\varepsilon]R_{\varphi_n}\xi_\psi)\\
&=(\xi_\psi|\mathbb{E}^{\varphi_n}[0,k_\varepsilon]R^2_{\varphi_n}\mathbb{E}^{\varphi_n}[0,k_\varepsilon]\xi_\psi)\\
&\le k_\varepsilon\cdot(\xi_\psi|\mathbb{E}^{\varphi_n}[0,k_\varepsilon]R_{\varphi_n}\mathbb{E}^{\varphi_n}[0,k_\varepsilon]\xi_\psi)\\
&=k_\varepsilon\cdot(\xi_\psi|R^{1/2}_{\varphi_n}\mathbb{E}^{\varphi_n}[0,k_\varepsilon]R^{1/2}_{\varphi_n}\xi_\psi)\\
&\le k_\varepsilon\cdot(\xi_\psi|R^{1/2}_{\varphi_n}\mathbb{E}^{\varphi_n}[0,+\infty)R^{1/2}_{\varphi_n}\xi_\psi)\\
&=k_\varepsilon\cdot(\xi_\psi|R_{\varphi_n}\xi_\psi)\\
&=k_\varepsilon\cdot(\xi_\psi|\xi_{\varphi_n})\\
\end{split}
\]
so that $0\le\lim_n (\xi_{\varphi_n}|\mathbb{E}^{\varphi_n}[0,k_\varepsilon]\xi_{\varphi_n})\le k_\varepsilon\cdot\lim_n (\xi_\psi|\xi_{\varphi_n})=0$. Thus $\lim_n \|\varphi_n\|_{M_{*+}}\le\varepsilon$ for all $\varepsilon>0$ and then $\lim_n \|\varphi_n\|_{M_{*+}}=0$.
\end{proof}

\subsection{Uniform integrability of sublevel sets of relative entropy}

Recall that relative entropy of a functional $\varphi\in M_{*+}$ with respect to a faithful $\psi\in M_{*+}$ is defined as
\begin{equation}
S(\varphi|\psi):=(\xi_\varphi|\ln R^2_\varphi\xi_\varphi)=(\xi_\psi |R^2_\varphi\ln R^2_\varphi\xi_\psi)\in [0,+\infty]
\end{equation}
(see \cite{3}, in which the notation for our $S(\varphi|\psi)$) is $S(\psi/\varphi)$, and \cite{20}). By convexity of the function $t\mapsto t\ln t$ and the bound $t\ln t\ge t-1$ for all $t>0$, one has
\[
S(\varphi|\psi)\ge \psi(1_M)\ln\psi(1_M)=\|\psi\|_{M_*}\cdot \ln\|\psi\|_{M_*}
\]
and
\[
S(\varphi|\psi)=(\xi_\psi |R^2_\varphi\ln R^2_\varphi\xi_\psi)\ge (\xi_\psi |(R^2_\varphi-1_M)\xi_\psi)=\varphi(1_M)-\psi(1_M)=\|\varphi\|_{M_*}-\|\psi\|_{M_*}.
\]
\begin{ex}
If $\varphi,\psi$ are states then $S(\varphi,\psi)\ge 0$. In the tracial case $\psi=\tau$, the relative entropy of the perturbed state $\varphi^h(\cdot):=\tau(h\cdot)$ of Example 3.2, is given by
\[
S(\varphi^h|\tau)=\tau(h\ln h)=\int_0^{+\infty}(\tau\circ\mathbb{E}^h)(d\lambda)\lambda\ln\lambda.
\]
\end{ex}
The following lower bound will be used below to estimate the degeneracy of ground states. Recall that the support $s_\varphi\in M_+$ of $\varphi\in M_{*+}$ is the projection onto the subspace $J\overline{M\xi_\varphi}\subseteq L^2(M,\tau)$.

\begin{lem}
Let $\varphi,\psi\in M_{*+}$ be states on $M$ and assume $\psi$ to be faithful. Then
\[
S(\varphi|\psi)\ge-\ln\psi(s_\varphi).
\]
\end{lem}
\begin{proof}
To shorten notation let us denote the support $s_\varphi$ just by $s$. Then $sR_\varphi(x\xi_\psi)=sJx^*\xi_\varphi=Jx^*Ja\xi_\varphi=Jx^*J\xi_\varphi=R_\varphi(x\xi_\psi)$ for all $x\in M$ so that $sR_\varphi=R_\varphi$ and $sR^2_\varphi=R^2_\varphi$. Moreover, $R_\varphi(sx\xi_\psi)=Jx^*s\xi_\varphi=Jx^*\xi_\psi=R_\varphi(x\psi)$ for all $x\in $ so that $R_\varphi\circ s=R_\varphi$ and $R_\varphi\circ s=R_\varphi$. Setting $\eta:=s\xi_\psi/\|s\xi_\psi\|$, consider the  probability measure $\mathbb{E}^\varphi_{\eta,\eta}(\cdot):=(\eta|\mathbb{E}^\varphi(\cdot)\eta)$. Since $R_\varphi s\xi_\psi=Js^*\xi_\varphi=Js\xi_\varphi=J\xi_\varphi=\xi_\varphi$, we have
\[
\begin{split}
S(\varphi|\psi)&=(\xi_\psi |R^2_\varphi\ln R^2_\varphi\xi_\psi)=\|s\xi_\psi\|^2\cdot (\eta |R^2_\varphi\ln R^2_\varphi \eta)=
\|s\xi_\psi\|^2\cdot\int_0^{+\infty}\mathbb{E}^\varphi_{\eta,\eta}(d\lambda)\,\lambda^2\ln\lambda^2\\
&\ge \|s\xi_\psi\|^2\cdot\ln\Bigl(\int_0^{+\infty}\mathbb{E}^\varphi_{\eta,\eta}(d\lambda)\,\lambda^2\Bigr)
=\ln\|R_\varphi \eta\|^2=\ln\Bigl(\|R_\varphi s\xi_\psi\|^2/\|s\xi_\psi\|^2\Bigr)\\
&=\ln\Bigl(\|\xi_\varphi\|^2/\|s\xi_\psi\|^2\Bigr)=\ln\Bigl(1/\|s\xi_\psi\|^2\Bigr)=-\ln(\xi_\psi|s\xi_\psi)=-\ln\psi(s).
\end{split}
\]
\end{proof}

\begin{thm}
A family $\mathcal{C}\subset M_{*+}$, whose relative entropies with respect to a fixed faithful $\psi\in M_{*+}$, are uniformly bounded above
\begin{equation}
\sup_{\varphi\in\mathcal{C}} S(\varphi|\psi)<+\infty,
\end{equation}
is bounded in $M_*$ and relatively uniformly integrable with respect to $\psi$.
\end{thm}

\begin{proof}
Since $\|\varphi\|_{M_*}-\|\psi\|_{M_*}\le S(\varphi|\psi)$, the boundedness of $\mathcal{C}$ follows from $\sup_{\varphi\in\mathcal{C}}\|\varphi\|_{M_*}\le \|\psi\|_{M_*}+\sup_{\varphi\in\mathcal{C}} S(\varphi|\psi)<+\infty$. Let now $\mathbb{E}^{R_\varphi}_\varphi$ be the measure on $[0,+\infty)$ obtained composing the spectral measure $\mathbb{E}^{R_\varphi}$ with the functional $\varphi\in\mathcal{C}$. For any $k\ge 1$, we have
\[
\begin{split}
(\xi_\varphi|\mathbb{E}^{R_\varphi}(k,+\infty)\xi_\varphi)&=\int_k^{+\infty}1\cdot\mathbb{E}^{R_\varphi}_\varphi(d\lambda)
\le \int_k^{+\infty}\frac{\ln\lambda^2}{\ln k^2}\cdot\mathbb{E}^{R_\varphi}_\varphi(d\lambda)\\
&\le \frac{1}{2\ln k}\cdot\int_0^{+\infty}\ln\lambda^2\cdot\mathbb{E}^{R_\varphi}_\varphi(d\lambda)= \frac{1}{2\ln k}\cdot(\xi_\varphi|\ln R^2_\varphi\xi_\varphi)= \frac{1}{2\ln k}\cdot S(\varphi|\psi)
\end{split}
\]
from which it follows that
\[
\lim_{k\to+\infty}\sup_{\varphi\in\mathcal{C}} (\xi_\varphi|\mathbb{E}^{R_\varphi}(k,+\infty)\xi_\varphi)\le\sup_{\varphi\in\mathcal{C}} S(\varphi|\psi)\cdot\lim_{k\to+\infty} 1/2\ln k=0.
\]
\end{proof}

\begin{thm}
Let $\mathcal{C}\subset M_{*+}$ be a family whose relative entropies with respect to a fixed faithful functional $\psi\in M_{*+}$, are uniformly bounded above
\[
\sup_{\varphi\in\mathcal{C}} S(\varphi|\psi)<+\infty.
\]
Let $\Xi_\mathcal{C}\subset L^2_+(M,\tau)$ be the family of positive vectors representing the functionals in $\mathcal{C}$,
\begin{equation}
\Xi_\mathcal{C}:=\{\xi\in L^2_+(M,\tau):\varphi(\cdot)=(\xi|\cdot\xi)\in\mathcal{C}\}
\end{equation}
Then the weak closure of $\Xi_\mathcal{C}$ in $L^2(M,\tau)$ does not contains zero.
\end{thm}

\begin{proof}
As in the proof of Theorem 3.9, $\sup_{\xi\in\Xi_\mathcal{C}}\|\xi\|^2\le \|\psi\|_{M_*}+\sup_{\varphi\in\mathcal{C}} S(\varphi|\psi)<+\infty$, so that $\Xi_\mathcal{C}$ is norm bounded,
hence weakly relatively compact in $L^2(M)$. If $\xi\in L^2_+(M)$ is a weak limit point of it in $L^2(M)$, there exists a sequence $\xi_n\in\Xi_\mathcal{C}$ weakly convergent to $\xi$. Since the relative entropy functional $\varphi\mapsto S(\varphi |\psi)$
is uniformly bounded on $\mathcal{C}$, by Theorem 3.9, the sequence $\{\varphi_{\xi_n}\in\mathcal{C}:\xi_n\in\Xi_\mathcal{C}\}$ is a fortiori relatively uniformly integrable.
If we had $\lim_n (\xi_{\varphi_n}|\xi_\psi)=0$, by Theorem 3.6, we would have $\lim_n\|\varphi_{\xi_n}\|_{M_*}=0$ which is a contradiction since $\lim_n\|\varphi_{\xi_n}\|_{M_*}=\lim_n 1=1$. Thus, necessarily, $(\xi_\psi|\xi)=\lim_n (\xi_{\varphi_n}|\xi_\psi)\neq 0$ and since $\xi_\psi\in L^2_+(M,\tau)$ is strictly positive as $\psi$ is faithful, we have $\xi\neq 0$.

%Since the relative entropy functional $\varphi\mapsto S(\varphi |\psi)$ is uniformly bounded on $\mathcal{C}$, this set is weakly relatively compact in $M_*$. $\varphi_n\in\mathcal{C}$ such that the sequence of representative vectors and such that the sequence $\varphi_n $weakly converges to some state $\varphi\in M_{*+}$. We want to prove that $\xi\neq 0$ showing that it is not orthogonal to the positive unit vector $\xi_\psi\in L^2_+(M)$ representing the state $\psi$. By contradiction assume that $0=(\xi_\psi|\xi)=\lim_n (\xi_{\varphi_n}|\xi_\psi)$ and let $\varphi\in M_{*+}$ be the weak limit of the sequence $\varphi_n\in M_{*+}$
%\[
%\varphi(x):=\lim_n\varphi_n(x)\qquad x\in M.
%\]
%Since the functional $\varphi$ is positive we have $\|\varphi\|_{M_*}=\varphi(1_M)=\lim_n\varphi_n(1_M)=\lim_n 1=1$ so that, in particular, $\varphi\neq 0$. However, by a lemma above,
%$\mathcal{C}$ and, a fortiori, the sequence $\varphi_n$, are uniformly integrable relatively to $\psi$. Thus, by another lemma above, we would have the contradiction $\varphi=0$.
\end{proof}

\section{Ground states by logarithmic Sobolev inequalities}

Let $\psi\in M_{*+}$ be a fixed, faithful, normal state on the Clifford algebra $M:=Cl(H,\mathcal{J})$.

\begin{defn}
An energy form $(\E,\F)$ on $L^2(M,\tau)$ satisfies a {\it logarithmic Sobolev inequality with respect to $\psi\in M_{*+}$} if, for some $\beta> 0$ and $\gamma\in\R$,
\begin{equation}
S(\varphi_\xi|\psi)\le \beta\cdot (\E[\xi]-\gamma)\qquad\xi\in L^2_+(M,\tau),\quad \|\xi\|=1.
\end{equation}
This implies that if $\xi\in\F\cap L^2_+(M,\tau)$ and $\|\xi\|=1$, then the relative entropy $S(\varphi_\xi|\psi)$ is finite and that (4.1) can be rewritten as the lower boundedness of the {\it free-energy type} functional
\begin{equation}
\gamma\le\E[\xi]-\frac{1}{\beta}S(\varphi_\xi|\psi)\qquad\xi\in\F\cap L^2_+(M,\tau),\quad \|\xi\|=1.
\end{equation}
\end{defn}
If $\gamma=\E[\xi_\psi]$, the LSI is equivalent to the lower boundedness of the variation of the energy form $\E$ with respect to the ground state by the relative entropy of the states
\begin{equation}
\frac{1}{\beta}\cdot S(\varphi_\xi|\psi)\le \E[\xi]-\E[\xi_\psi]\qquad\xi\in\F\cap L^2_+(M,\tau),\quad \|\xi\|=1.
\end{equation}

\begin{ex}(Logarithmic Sobolev inequality for the Clifford-Dirichlet form)
The Clifford-Dirichlet satisfies the following logarithmic Sobolev inequality
\begin{equation}
S(\varphi_\xi|\tau)\le\E_N[\xi]\qquad\xi\in L^2_+(M,\tau),\,\, \|\xi\|=1
\end{equation}
with parameters $\beta=1$ and $\gamma=0$. The inequality was proved in \cite{14} with $\beta=\sqrt{3}$ and the best value $\beta=1$, conjectured by L. Gross, was found in \cite{5}.
\end{ex}

We may now show that a logarithmic Sobolev inequality directly implies existence and finite degeneracy of the ground state.

\begin{thm}[Existence and degeneracy of ground states of energy forms]
Let $(\E,\F)$ be an energy form on $L^2(M,\tau)$, satisfying the logarithmic Sobolev inequality (4.1) with respect to a faithful state $\psi\in M_{*+}$, and let $(H,D(H))$ be the associated self-adjoint operator.\\
Then the lowest energy $\lambda_0:=\inf sp(H)$ is an eigenvalue of finite multiplicity $m_0\le e^{\beta(\lambda_0-\gamma)}$.
\end{thm}

\begin{proof}
As $\beta(\E[\cdot]-\gamma\|\cdot\|^2)$ is an energy form for any $\beta>0,\gamma\le\lambda_0$, we may suppose $\beta=1, \gamma=0$ and deal with the logarithmic Sobolev inequality
\[
S(\varphi_\xi|\psi)\le \E[\xi]\qquad\xi\in\F\cap L^2_+(M,\tau),\quad \|\xi\|=1,
\]
which implies, in particular, $sp(H)\subseteq [0,+\infty)$, i.e.  $\lambda_0\ge 0$. Since $\sqrt{\lambda_0}\in sp(\sqrt{H})$, we may consider a Weyl normalized sequence $\xi_n\in D(H)=\F$ such that
$\lim_n\|\sqrt{H}\xi_n-\sqrt{\lambda_0}\xi_n\|=0$. It follows that the sequence minimises the energy form, as
\[
\begin{split}
\sup_n\E[\xi_n]&=\sup_n\|\sqrt{H}\xi_n\|^2=\sup_n\|\sqrt{H}\xi_n-\sqrt{\lambda_0}\xi_n+\sqrt{\lambda_0}\cdot\xi_n\|^2\\
&\le2\sup_n\Bigl(\|\sqrt{H}\xi_n-\sqrt{\lambda_0}\xi_n\|^2 + \|\sqrt{\lambda_0}\cdot\xi_n\|^2\Bigr)\\
&\le2\Bigl(\sup_n\|\sqrt{H}\xi_n-\sqrt{\lambda_0}\xi_n\|^2 +\lambda_0\Bigr)<+\infty
\end{split}
\]
and
\[
\begin{split}
\lim_n(\E[\xi_n]-\lambda_0)&=\lim_n(\|\sqrt{H}\xi_n\|^2-\|\sqrt{\lambda_0}\xi_n\|^2)=\lim_n(\sqrt{H}\xi_n-\sqrt{\lambda_0}\xi_n|\sqrt{H}\xi_n+\sqrt{\lambda_0}\xi_n)\\
&\le\lim_n\|\sqrt{H}\xi_n-\sqrt{\lambda_0}\xi_n\|\cdot \|\sqrt{H}\xi_n+\sqrt{\lambda_0}\xi_n\|\\
&\le(\sup_n\sqrt{E[\xi_n]}+\lambda_0)\lim_n\|\sqrt{H}\xi_n-\sqrt{\lambda_0}\xi_n\|=0.
\end{split}
\]
By the first Beurling-Deny contraction property of the energy form $(\E,\F)$, the normalized sequence $|\xi_n|\in D(H)=\F$ is still a Weyl approximating sequence for $\sqrt{\lambda_0}$
\[
\begin{split}
\lim_n \|\sqrt{H}|\xi_n|-\sqrt{\lambda_0}|\xi_n|\|^2&=\lim_n \{\E[|\xi_n|]+\lambda_0-2\sqrt{\lambda_0}(\sqrt{H}|\xi_n|||\xi_n|)\}\\
&\le\lim_n \{\E[|\xi_n|]+\lambda_0-2\sqrt{\lambda_0}\sqrt{\lambda_0}\}\\
&\le \lim_n\{\E[\xi_n]-\lambda_0\}=0.
\end{split}
\]
Hence, we may assume that the sequence is made by positive, unit vectors $\xi_n\in\F\cap L^2_+(M,\tau)$ and consider a weak limit of it $\xi\in L^2_+(M,\tau)$.
%Since the energy form is convex and lower semicontinuous with respect to the norm topology, it is also lower semicontinuous with respect to the weak topology of $L^2(M,\tau)$.
Since $\lim_n\|\sqrt{H}\xi_n-\sqrt{\lambda_0}\xi_n\|=0$ implies that $\lim_n (\sqrt{H}\xi_n-\sqrt{\lambda_0}\xi_n)=0$ weakly in $L^2(M,\tau)$, we have $\lim_n\sqrt{H}\xi_n=\lim_n(\sqrt{H}\xi_n-\sqrt{\lambda_0}\cdot\xi_n)+\lim_n\sqrt{\lambda_0}\cdot\xi_n=\sqrt{\lambda_0}\cdot\xi$ weakly in $L^2(M,\tau)$.
As the operator $(\sqrt{H},\F)$ is closed with respect to the norm topology of $L^2(M,\tau)$, by Mazur's Theorem, it is also closed with respect to the weak topology of $L^2(M,\tau)$ and we may conclude that $\xi\in D(\sqrt{H})=\F$ and $\sqrt{H}\xi=\lambda_0\cdot\xi$ and consequently that $\xi\in D(H)$ and $H\xi=\lambda_0\cdot\xi$. What is left to be proved to deduce that $\lambda_0$ is an eigenvalue of $H$ is that $\xi\neq 0$. This follows from Theorem 3.10 and the fact that by (4.1), the relative entropy $S(\cdot|\psi)$ is uniformly bounded above on the set of states $\mathcal{C}:=\{\varphi_{\xi_n}\in M_{*+}:n\in\N\}$
\[
\sup_{\varphi\in\mathcal{C}} S(\varphi|\psi)=\sup_n S(\varphi_{\xi_n}|\psi)\le\beta(\sup_n\E[\xi_n]-\gamma)<+\infty.
\]
To prove the finite degeneracy of $\lambda_0$, we start to notice that, since the energy form is real and has the first Beurling-Deny contraction property, the eigenspace $E_0\subset\F$ corresponding to $\lambda_0$ has a basis of positive vectors. Let $m\in\N^*$ be a non zero natural number and $\xi_i\in\F\cap L^2_+(M,\tau)$, $i=1,\dots,m$, a family of positive, orthonormal vectors such that
$\E[\xi_i]=\lambda_0$. Denote by $s_i\in M$ the support projections of the states $\varphi_{\xi_i}(\cdot):=(\xi_i|\cdot\xi_i)$, $i=1,\dots,m$.  By a lemma above, we have
\[
S(\varphi_i|\psi)\ge -\ln\psi(s_i)\qquad i=1,\dots,m,
\]
and, by the logarithmic Sobolev inequality, we have too
\[
-\ln\psi(s_i)\le\beta(\lambda_0-\gamma)\qquad i=1,\dots,m
\]
and we get the bound $\psi(s_i)\ge e^{-\beta(\lambda_0-\gamma)}$ for all $i=1,\dots,M$. By the orthogonality of the eigenvectors, the support projections of the states are mutually orthogonal so that
\[
1=\psi(1_M)\ge \psi(s_1)+\cdots +\psi (s_m)\ge m\cdot e^{-\beta(\lambda_0-\gamma)}
\]
and the multiplicity $m_0$ of $\lambda_0$ is bounded by $1\le m_0\le e^{\beta(\lambda_0-\gamma)}$.
\end{proof}

\begin{cor}
The ground state $\psi_0$ is non degenerate if $\beta(\lambda_0-\gamma)<\ln 2$ and in this case $(\E_0,\F_0):=(\E[\cdot]-\lambda_0\cdot\|\cdot\|^2,\F)$ is a Dirichlet form with respect to $\psi_0$.
 \end{cor}
 \begin{proof}
 If $\beta(\lambda_0-\gamma)<\ln 2$, then $1\le m_0<2$ so that $m_0=1$. Since $(\E_0,\F_0)$ is an energy form and $\E_0[\xi_{\psi_0}]=\E[\xi_{\psi_0}]-\lambda_0\|\xi_{\psi_0}\|^2=0$, it follows from Proposition 4.10 ii) in \cite{6} that $(\E_0,\F_0)$ is a Dirichlet form with respect to $\psi_0$.
 \end{proof}

\section{Perturbations and stability of ground states}

In this section we consider the stability of a ground state, supposed to be unique. We first show that logarithmic Sobolev inequalities persist under certain perturbations of an energy form and then derive that the perturbed energy form still admits a unique ground state together with a lower bound upon the energy variation. It is worth noticing here that in \cite{15} the author proves the invariance of the intrinsic hypercontractivity property
of Schr\"odinger type operators on Riemannian manifolds with a class of potentials including the bounded ones, by using the equivalence of hyperconractivity and logarithmic Sobolev inequalities.
\vskip0.2truecm\noindent
Here $\Sigma_*(M)$ denotes the set of normal states of $M$.
\vskip0.2truecm\noindent
If $(h,D(h))$ is a lower bounded, self-adjoint operator affiliated to $M$ and $\varphi\in\Sigma_*(M)$, we denote by
$\mathbb{E}^h_\varphi$ the probability measure on $\R$ obtained as follows
\[
\mathbb{E}^h_\varphi(\Omega):=\varphi(\mathbb{E}^h(\Omega))\qquad \Omega\subseteq\R\,\,\text{Borel set},
\]
and we define $\varphi(h):=\int_\R\mathbb{E}^h_\varphi (d\lambda)\lambda\in\R\cup\{+\infty\}$. This value is finite in case $\xi_{\psi}$ belongs to the domain of $h$ acting (by left multiplication) on $L^2(M,\tau)$.
\vskip0.2truecm\noindent
In this section we fix a faithful, normal state $\psi\in\Sigma_*(M)$ and a parameter $\beta>0$. We will make use of the following

\begin{lem} (\cite{20} Lemma 12.1)
Let $(h,D(h))$ be a lower bounded, self-adjoint operator affiliated to $M$ such that
\[
c_\beta (h,\psi):=\inf\Big\{\frac{1}{\beta}\cdot S(\varphi|\psi)+\varphi(h):\varphi\in\Sigma_*(M)\Big\}\in\R.
\]
Then, there exists a unique, faithful, normal state $\psi^h_\beta\in\Sigma_*(M)$ such that
\[
c_\beta (h,\psi)=\frac{1}{\beta}\cdot S(\psi^h_\beta|\psi)+\psi^h_\beta (h).
\]
Moreover
\[
\frac{1}{\beta}\cdot S(\varphi|\psi^h_\beta)+c_\beta(h,\psi)\le\frac{1}{\beta}\cdot S(\varphi|\psi)+\varphi(h)\qquad \varphi\in\Sigma_*(M).
\]
The condition $c_\beta (h,\psi)\in\R$ is verified if $\psi(h)<+\infty$.
\end{lem}
\begin{proof}
See Theorem 5.26 and Chapter 12 page 213, Lemma 12.1 and formula (12.9) of \cite{20}.
\end{proof}
The state $\psi^h_\beta\in\Sigma_*(M)$ is called {\it the perturbed state}.

\begin{ex}
If $\psi$ is the trace state $\tau$, then the perturbed state is
\[
\psi^h_\beta(\cdot)=\tau(e^{-\beta h}\cdot)/\tau(e^{-\beta h}).
\]
and it is easy to see that $c_\beta(h,\tau)=$ equals a finite temperature partition function
\[
c_\beta(h,\tau)=-\frac{1}{\beta}\cdot\ln\tau(e^{-\beta h}).
\]
\end{ex}

\begin{thm}
Let $(\E_0,\F_0)$ be an energy form on $L^2(M,\tau)$, satisfying a logarithmic Sobolev inequality with respect to $\psi_0\in\Sigma_*(M)$ for some $\beta>0$ and $\gamma_0\in\R$
\begin{equation}
S(\varphi_\xi|\psi_0)\le\beta\bigl(\E_0[\xi]-\gamma_0\bigr)\qquad \xi\in L^2_+(M,\tau),\,\, \|\xi\|=1.
\end{equation}
Let $(h,D(h))$ be a lower bounded, self-adjoint operator affiliated to $M$ such that
\begin{equation}
c_\beta (h,\psi_0)\in\R
\end{equation}
and let $\psi^h_\beta\in\Sigma_*(M)$ be the perturbed state. Then the perturbed quadratic form
\begin{equation}
\F_h:=\F_0\cap D(h)\cap JD(h)\qquad \E_h[\xi]:=\E_0[\xi]+(\xi|(h+JhJ)\xi)/2
\end{equation}
is real, it verifies the first Beurling-Deny contraction property and it satisfies the perturbed logarithmic Sobolev inequality w.r.t. the perturbed state $\psi^h_\beta\in\Sigma_*(M)$ with
$\gamma_h:=\gamma_0+c_\beta(h,\psi_0)$
\begin{equation}
S(\varphi_\xi|\psi^h_\beta)\le\beta\bigl(\E_h[\xi]-\gamma_h\bigr)\qquad \xi\in\F_h\cap L^2_+(M,\tau),\,\, \|\xi\|=1.
\end{equation}
 If $(\E_h,\F_h)$ is densely defined and closable the same holds true for its form closure, the perturbed ground state energy $\lambda_h:=\inf\{\E_h[\xi]\in\R:\xi\in\F_h\}$
 is bounded below as follows
 \begin{equation}
-\infty<\gamma_0+c_\beta (h,\psi_0)=\gamma_0+\frac{1}{\beta}\cdot S(\psi^h_\beta|\psi_0)+\psi^h_\beta (h)\le\lambda_h
\end{equation}
and the variation of the ground state energies between perturbed and unperturbed systems, is bounded below as follows
\begin{equation}
-\infty<(\gamma_0-\lambda_0)+\frac{1}{\beta}\cdot S(\psi^h_\beta|\psi_0)+\psi^h_\beta (h)\le\lambda_h-\lambda_0.
\end{equation}
\end{thm}

\begin{proof}
This is a consequence of the perturbed logarithmic Sobolev inequality, the  positivity of the relative entropy, the definition of $\gamma_h$ and the hypothesis $c_\beta(h,\psi_0)\in\R$.
\end{proof}
{\it It should be stressed that, in general, the perturbed state is by no means the ground state of the perturbed Hamiltonian}.
\begin{cor}
If $\psi_0$ is the trace state $\tau$ then
\begin{equation}
(\gamma_0-\lambda_0)-\frac{1}{\beta}\cdot\ln\tau(e^{-\beta h}).\le\lambda_h-\lambda_0.
\end{equation}
\end{cor}

\begin{rem}
i) The lower boundedness of $(h,D(h))$ guarantees that $c_\beta (h,\psi)\in\R\cup\{+\infty\}$ which is the essential hypotheses by which the perturbed state exists. The hypotheses $c_\beta (h,\psi_0)<+\infty$ is needed to avoid triviality, as it just means that for at least one state $\varphi\in\Sigma_*(M)$ both the relative entropy $S(\varphi|\psi)$ and the energy $\varphi(h)$ are finite.
\vskip0.2truecm\noindent
ii) A natural condition by which $(\E_h,\F_h)$ is densely defined, closable is in Theorem 6.4 of \cite{6}.
\vskip0.2truecm\noindent
iii) The bound (5.7) can be compared to the bound (3.2) of Lemma 4.1 in \cite{13} (valid, however, in general probability gage spaces and for perturbations not necessarily bounded below). One may conjecture that (5.7) could be valid under the solely finiteness of $\tau(e^{-\beta h})$.
\end{rem}

\section{On existence and uniqueness of the ground state of physical Hamiltonians}
In \cite{13} Leonard Gross obtained the existence and uniqueness of the ground state of physical Hamiltonians describing spin $1/2$ Dirac particles in a four dimensional space-time, interacting with an external neutral scalar field (Theorem 7 in \cite{13}).
In that framework $\calH$ is the Hilbert space direct sum of the irreducible representation spaces of an electron and a positron and the conjugation $\J$ on it is the one that by second quantization gives rise to the the Parity-Charge-Time PCT symmetry. The proof relies in showing that
\begin{itemize}
\item the Hamiltonian of the interacting system splits as $H=H_0+H_1$ (formula (5.39) of \cite{13}), where $H_1$, which depends on the interaction only, has the form (in our notations)
\[
H_1=(h+JhJ)/2
\]
for a self-adjoint $h\in Cl(\calH,\J)$ (formula (5.33) of \cite{13} and related  comments)
\item $H_0$ is the sum of the second quantized free Hamiltonian of the spin 1/2 Dirac particles, a trace class operator on $L^2(Cl(\calH,\J),\tau)$ and a multiple of the number operator $N$
\item the semigroup $e^{-tN}$ on $L^2(Cl(\calH),\tau)$ is hypercontractive (Lemma 6.1 of \cite{13}) and its ground state is nondegenerate
\item the semigroup $e^{-tH_0}$ on $L^2(Cl(\calH),\tau)$ is hypercontractive with a nondegenerate ground state too by  Theorem 6 of \cite{13}, since $H_0$ dominates a positive multiple of $N$
\item $H$ admits a nondegenerate ground state by  Theorem 4 and Lemma 4.1 of \cite{13}.
\end{itemize}

By the analysis developed in the previous sections, we provide an {\it infinitesimal} proof of the main result of Theorem 7 of \cite{13}, based, {\it in primis}, on the logarithmic Sobolev inequality (4.4) satisfied by the number operator $N$ on a the Clifford algebra $M=Cl(\calH,\J)$.

\begin{thm}
Let $H$ be the Hamiltonian of a system of second quantized spin 1/2 Dirac particles, interacting with an external neutral scalar field and subject to the hypotheses of Theorem 7, discussed in Section 5 of \cite{13}.\\
Then there exists a strictly positive eigenvector $\psi^H\in L^2_+(Cl(\calH,\J),\tau)$ of $H$ corresponding to the non degenerate lowest energy level $\inf sp(H)$.
\end{thm}

\begin{proof}
On the GNS space $L^2(Cl(\calH,\J),\tau)$, the Hamiltonian $H$ of the system may be written
\[
H=D^{-1}d\Gamma(A)D+L_\alpha+R_\alpha
\]
where i) $A=G+F+(\delta m)I_\calH$ is a self adjoint operator on $\calH$, commuting with $\J$, sum of the one-particle free Hamiltonian $G$, a trace-class operator $F$ and a multiple
$(\delta m)I_\calH$ of the identity operator ii) $\alpha^*=\alpha\in Cl(\calH,\J)$ (the term $L_\alpha+R_\alpha$ carries all the pair-creation and pair-annihilation contribution of the interacting part of
$H$).
\vskip0.2truecm\noindent
Since $A\ge\mu\cdot I_\calH$ for some $\mu>0$, by Lemma 6.2, it turns out that $D^{-1}d\Gamma(A)D$ generates a Markovian semigroup and a Dirichlet form $\E_0$ with respect to $\tau$.
\vskip0.2truecm\noindent
Since, moreover, $D^{-1}d\Gamma(A)D\ge\mu\cdot N$ for a suitable $\mu>0$, by the logarithmic Sobolev inequality (4.4) satisfied by the Clifford Dirichlet form $\E_N$, we get a logarithmic Sobolev inequality satisfied by $\E_0$
\begin{equation}
S(\varphi_\xi|\tau)\le\frac{1}{\mu}\cdot \E_0[\xi]\qquad\xi\in L^2_+(M,\tau),\,\, \|\xi\|=1
\end{equation}
with parameters $\beta=1/\mu$ and $\gamma_0=0$. Setting $h:=2\alpha\in Cl(\calH,\J)$ , one has $L_\alpha+R_\alpha=(h+JhJ)/2$ so that, applying Theorem 5.3 above, we get that the quadratic form $\E_h=\E_0+(\cdot|(h+JhJ)\cdot)/2$ of $H$, which is closed since $h$ is bounded, satisfies the logarithmic Sobolev inequality
\begin{equation}
S(\varphi_\xi|\psi^h_{\mu^{-1}})\le\mu^{-1}\bigl(\E_h[\xi]-\gamma_h\bigr)\qquad \xi\in\F_h\cap L^2_+(M,\tau),\,\, \|\xi\|=1.
\end{equation}
Theorem 4.3 above then implies that $H$ has a strictly positive eigenvector corresponding to the eigenvalue $\inf sp(H)$.
\end{proof}

\subsection*{Acknowledgement} The author warmly thanks the anonymous referee for his/her careful and scrupulous reading and suggestions.
%\newpage

%------------------------------------------------------------------------------------------------------------------------------------------------
\end{document}